\documentclass[11pt]{article}
\usepackage{amsmath}
\usepackage{amsthm}
\usepackage{amsfonts}
\usepackage{amssymb}
\usepackage{graphicx}
\usepackage{epstopdf}
\usepackage{caption}
\usepackage{listings}
\usepackage{paralist}
\usepackage{tikz}
\usepackage[ruled,commentsnumbered,linesnumbered]{algorithm2e}
\usepackage{lipsum}
\usepackage{hyperref}
\usepackage{url}
\usepackage{multicol}
\usepackage{blindtext}
\usepackage{array}
\usepackage{floatrow}
\usepackage{caption}
\usepackage{lipsum}
\usepackage{wrapfig}
\usepackage{grffile}
\usepackage{xcolor}
\usepackage{color}
\usepackage{algorithmic}
\usepackage{booktabs}
\usepackage{multirow}
\usepackage{fullpage}
\usepackage{rotating}

\hypersetup{
    bookmarks=true,         
    unicode=false,          
    pdftoolbar=true,        
    pdfmenubar=true,        
    pdffitwindow=false,     
    pdfstartview={FitH},    
    pdftitle={My title},    
    pdfauthor={Author},     
    pdfsubject={Subject},   
    pdfcreator={Creator},   
    pdfproducer={Producer}, 
    pdfkeywords={keyword1} {key2} {key3}, 
    pdfnewwindow=true,      
    colorlinks=true,       
    linkcolor=blue,          
    citecolor=blue,        
    filecolor=magenta,      
    urlcolor=blue           
}

\input{prelude1}
\newfont{\mycrnotice}{ptmr8t at 7pt}
\newfont{\myconfname}{ptmri8t at 7pt}

\newcommand{\unsafe}{\num U}

\newcommand{\findldf}{\mathit{ComputeLDF}}

\theoremstyle{plain}
\newtheorem{theorem}{Theorem}[section]
\theoremstyle{definition}
\newtheorem{definition}[theorem]{Definition}
\newtheorem{lemma}[theorem]{Lemma}

\newtheorem{corollary}[theorem]{Corollary}

\title{Bounded Verification with On-the-Fly Discrepancy Computation\thanks{We gratefully acknowledge the feedback from anonymous referees on a previous draft of this technical report. The results presented here came about from work supported and funded by the National Science Foundation (grant: CAR 1054247 and NSF CSR 1016791) and the Air Force Office of Scientific Research (AFOSR YIP FA9550-12-1-0336).}}
\author{Chuchu Fan  and Sayan Mitra \\ 
{\sf \{cfan10, mitras\}@illinois.edu} \\
Coordinated Science Laboratory\\ 
University of Illinois at Urbana Champaign \\ 
Urbana, IL 61801}

\begin{document}
\maketitle
\begin{abstract}
Simulation-based verification algorithms can provide formal safety guarantees for nonlinear and hybrid systems. The previous algorithms rely on user provided model annotations called discrepancy function, which are crucial for computing reachtubes from simulations. In this paper, we eliminate this requirement by presenting an algorithm for computing piece-wise exponential discrepancy functions. The algorithm relies on computing local convergence or divergence rates of trajectories along a simulation using a coarse over-approximation of the reach set and bounding the maximal eigenvalue of the Jacobian over this over-approximation. The resulting discrepancy function preserves the soundness and the relative completeness of the verification algorithm. We also provide a coordinate transformation method to improve the local estimates for the convergence or divergence rates in practical examples. We extend the method to get the input-to-state discrepancy of nonlinear dynamical systems which can be used for compositional analysis. Our experiments show that the approach is effective in terms of running time for several benchmark problems, scales reasonably to larger dimensional systems, and compares favorably with respect to available tools for nonlinear models.
\end{abstract}

\section{Introduction}
\label{sec:intro}

Verifying and falsifying nonlinear, switched, and hybrid system models using numerical simulations have been studied in detail~\cite{DMVemsoft2013,HuangM:HSCC2014,annpureddy2011s,girard2010approximately,donze2007systematic}. The bounded time safety verification problem for a given model is parameterized by a time bound, a set of initial states, and a set of unsafe states and it requires one to decide if there exists a behavior of the model that reaches any unsafe set from any initial state. The simulation-based procedure for this problem first generates a set of numerical approximations of the behaviors from a finite sampling of the initial states. Next, by bloating these simulations by an appropriately large factor it computes an over-approximation of the reachable states from the initial set. If this over-approximation  proves safety or produces a counter-example, then the algorithm decides, otherwise, it draws more samples of initial states and repeats the earlier steps to compute more precise over-approximation.
With post-processing of the reachtube over-approximations this basic procedure can be utilized to verify termporal precedence~\cite{DuggiralaWMVM14} and richer classes of properties~\cite{dang2008sensitive}. 

In order to make this type of procedure sound, the bloating factor should be chosen to be large.
Specifically, it should be large enough to make each bloated simulation an over approximation of the reachable states of the system not only from the sampled initial state, but also from a large enough neighborhood of that state so that the union of these neighborhoods cover the entire set of initial states. 
On the other hand, to make the procedure complete, or at least relatively complete modulo the precision of the machine, it should be possible to make the error due to bloating arbitrarily small for any point in time.
These two opposing requirements are captured in the definition of a {\em discrepancy function\/} of~\cite{DMVemsoft2013}: For an $n$-dimensional dynamical system, it is any function $\beta:\reals^{2n} \times \nnreals \rightarrow \nnreals$, such that (a) it gives an upper-bound on the distance between any two trajectories $\xi(x,t)$  and $\xi'(x,t)$ of the system  $|\xi(x,t) - \xi(x',t)| \leq \beta(x,x',t)$, and (b) it vanishes as $x$ approaches $x'$.  
Simply using the Lipschitz constant of the dynamic function gives one such bound, but it grows exponentially with time even for some incrementally stable models~\cite{angeli2002lyapunov}.

In~\cite{DMVemsoft2013}, it is observed that the notion of a contraction metric~\cite{slotine98} gives a much tighter bound and it provided heuristics for finding them for some classes of polynomial systems.
Sensitivity analysis approach gives strict error bounds for linear systems~\cite{donze2007systematic}, but for nonlinear models the bounds are less clear. We present a more detailed overview of related work in Section~\ref{sec:related}.
This paper fills this gap by providing a subroutine that computes  a local version of the discrepancy function which turns out to be adequate and effective for sound and relatively complete simulation-based verification. 
This subroutine, $\findldf$, itself uses a Lipschitz constant and the Jacobian of the dynamic function (the right hand side of the differential equation)  and simulations of the system. The Lispchitz constant is used to construct a coarse, one-step over-approximation of the reach set of the system along a simulation. Then it computes an upper bound on the maximum eigenvalue of the symmetric part of the Jacobian over this over approximation, using a theorem from matrix perturbation theory.
This gives an exponential bound on the distance between two trajectories, but roughly, the exponent is the best it can be as it is close to the maximum eigenvalue of the linear approximation of the system in the neighborhood.

We propose two practical extensions of this approach. 
First, we show that a linear coordinate transformation can bring about exponential improvements in the estimated distance. Secondly, we propose a technique for computing input-to-state discrepancy functions for analyzing composed systems and systems with bounded nondeterministic inputs. 
%
We report the results from a number of experiments performed with a prototype implementation of this approach applied to safety verification.



\section{Background}
\label{sec:prelims}

\subsection{Notations}
For a vector $x \in \reals^n$, $\|x\|$ is the $l^2$-norm of $x$
and $x_i$ denotes its $i^{th}$ component.
For $\delta \geq 0$, $B_\delta(x) = \{x' \in \reals^n  \ | \ ||x' - x || \leq \delta \}.$
For a set $S \subseteq \reals^n$, $B_{\delta}(S) = \cup_{x \in S}B_\delta(x)$.
Let $S \oplus B_{\delta}(0)$ represents the Minkowski sum of $S$ and $B_{\delta}(0)$.
Therefore,  $S \oplus B_\delta(0) = B_\delta(S)$. 
For sets $S_1,S_2 \subseteq \reals^n$, $hull(S_1,S_2)$ is their convex hull.
The diameter of a compact set $S$ is $dia(S) = \sup_{x_1,x_2 \in S}\|x_1-x_2\|$.

A continuous function $f : \reals^n \rightarrow \reals$ is {\em smooth\/} if all its higher derivatives and partial derivatives exist and are also continuous. It has a Lipschitz constant $L\geq 0$ if for every  $x,x' \in \reals^n$, $||f(x) - f(x')|| \leq L||x-x'||$.
%
%
A function $f:\nnreals \rightarrow \nnreals$ is a {\em class $\K$ function\/} if it is continuous, strictly increasing, and $f(0)=0$. 

We denote the transpose of a matrix $A$ by $A^T$. 
The {\em conjugated transpose\/} of $A$ is the matrix $A^H$ obtained by  replacing each entry in $A^T$  with its complex conjugate.

Given a differentiable vector-valued function $f: \mathbb{R}^n \times \mathbb{R}^{\geq 0} \rightarrow \mathbb{R}^n$, the {\em Jacobian ${J}_f$\/} of $f$ is the matrix-valued function of all the first-order partial derivatives of $f$. Let $f_i, i=1\dots n: \reals^n \rightarrow \nnreals$ be the scalar components of $f$. The Jacobian of $f$ is: $(J_{f}(x))_{ij} = \frac{\partial f_i(x)}{\partial x_j}$.
The {\em symmetric part of the Jacobian of $f$\/} matrix is defined as $\frac{1}{2}({J_f(x)}+{J_f}^T(x))$.

%

For an $n \times n$ matrix $A$, $\|A\|$ represents the $l^2$-norm of $A$: $\|A\| = \sqrt{\lambda_{\textrm{max}}(A^HA)}$. If $\forall x \in \reals^n$, $x^TAx \leq 0$, then we say $A$ is negative-semidefinite, and write $A \preceq 0$. We write $A \preceq B$ if $A-B \preceq 0$.

\subsection{Safety Verification Problem}
\label{sec:dynamical}

Consider an $n$-dimensional {\em autonomous dynamical system\/}:
\begin{eqnarray} \label{eq:dynamic system}
\dot{x} = f(x),
\label{eq:system}
\end{eqnarray}
where $f:\reals^n \rightarrow \reals^n$ is a Lipschitz continuous function.
A {\em solution or a trajectory\/} of the system is a function $\xi:\reals^n \times \nnreals \rightarrow \reals^n$ such that for any initial point $x_0 \in \reals^n$ and at any time $t >0$,  $\xi(x_0, t)$ satisfies the differential equation~(\ref{eq:system}).

The {\em bounded-time safety verification problem\/} is parameterized by:
\begin{inparaenum}[(a)]
\item an $n$-dimensional dynamical system, that is, the function $f$ defining the right hand side of its differential equation,
\item a compact set $\Theta \subseteq \reals^n$  of initial states,
\item an open set $\unsafe \subseteq \reals^n$ of unsafe states, and
\item a time bound $T > 0$.
\end{inparaenum}
A state $x$ in $\reals^n$ is {\em reachable from $\Theta$ within a time interval $[t_1, t_2]$\/} if there exists an initial state $x_0 \in \Theta$ and a time $t \in [t_1,t_2]$ such that $x =\xi(x_0,t)$. The set of all reachable states in the interval $[t_1,t_2]$ is denoted by $\reach{}(\Theta,[t_1,t_2])$. If $t_1 =0$ then we write $\reach{}(t_2)$ when set $\Theta$ is clear from the context.
Given a bounded-time safety verification problem, we would like to design algorithms for deciding if any reachable state is safe, that is, if $\reach{}(T) \cap \unsafe = \emptyset$. If there exists some $\epsilon>0$ such that $B_{\epsilon}(\reach{}(T)) \cap \unsafe = \emptyset$, we say the system is robustly safe. 
A sequence of papers~\cite{DMVemsoft2013,DuggiralaWMVM14,donze2007systematic} presented algorithms for solving this problem for a broad class of nonlinear dynamical, switched, and hybrid systems. In the remainder of this section, we present an overview of this approach. (Figure~\ref{fig:algo1}).

\subsection{Simulations, Reachtubes and Annotations}
\label{ssec:simann}
The algorithm uses simulation oracles that give sampled numerical simulations of the system from individual initial states.
%
\begin{definition} \label{def:simulation}
A {\em $(x_0, \tau, \epsilon,T)$-simulation\/} of the system described in Equation~(\ref{eq:system}) is a sequence of time-stamped sets  $(R_0, t_0)$, $(R_1,t_1) \ldots, (R_n,t_n)$ satisfying:
\begin{enumerate}[(1)]
\item Each $R_i$ is a compact set in $\reals^n$ with $\mathit{dia}(R_i) \leq \epsilon$.
\item The last time $t_n = T$ and for each $i$, $0 < t_i - t_{i-1} \leq \tau$, where the parameter $\tau$ is  called the {\em sampling period\/}.
\item For each $t_i$, the trajectory from $x_0$ at $t_i$ is in $R_i$, i.e., $\xi(x_0,t_i) \in R_i$, and 
for any $t \in [t_{i-1}, t_i]$, the solution $\xi(x_0,t) \in hull(R_{i-1},R_i)$.
\end{enumerate}
\end{definition}
Simulation engines generate a sequence of states and error bounds using numerical integration.
Libraries like CAPD~\cite{capd} and VNODE-LP~\cite{vnode2006} compute such simulations
for a wide range of nonlinear dynamical system models and 
the $R_i$'s are represented by some data structure like hyperrectangles.

Closely related to simulations are {\em reachtubes\/}.
For a set of states $D\subseteq \reals^n$,
a {\em $(D, \tau, T)$-reachtube\/} of~(\ref{eq:system}) is a sequence of time-stamped sets $(R_0, 0), (R_1,t_1) \ldots, (R_n,t_n)$ satisfying:
\begin{enumerate}[(1)]
\item Each $R_i \subseteq \reals^n$ is a compact set of states.
\item The last time $t_n = T$ and for each $i$, $0 \leq t_i - t_{i-1} \leq \tau$.
\item For any $x_0 \in D$, and any time $t \in [t_{i-1}, t_i]$, the solution $\xi(x_0,t) \in R_i$.
\end{enumerate}
A reachtube is analogous to a simulation from a set of states, but they are much harder to compute.
In fact, an algorithm for computing exact reachtubes readily solves the safety verification problem.

The algorithms in~\cite{DMVemsoft2013,HuangFMMK14} require the user to decorate the model with annotations called {\em discrepancy functions\/} for computing reachtubes.
\begin{definition}
\label{def:disc}
A continuous function $\beta:\reals^n \times \reals^n \times \nnreals \rightarrow \nnreals$ is a {\em discrepancy function\/} of the system in Equation~(\ref{eq:system}) if
\begin{enumerate}[(1)]
\item for any pair of states $x, x' \in \reals^n$, and any time $t >0$,
\begin{eqnarray}
\|\xi(x,t) - \xi(x',t)\| \leq \beta(x,x',t), \mbox{and}
\label{eq:df1}
\end{eqnarray}
\item for any $t$, as $x \rightarrow x'$, $\beta(.,.,t) \rightarrow 0$,
\chuchu{\item $\forall \epsilon>0,  \forall x,x' \in \reals^n, \exists \delta$ such that for any time $t, \|x-x'\|<\delta \Rightarrow \beta(x,x',t)<\epsilon$.}
\end{enumerate}
\end{definition}

If the function $\beta$ meets the two conditions for any pair of states $x,x'$ in a compact set $K$ \chuchu{in both condition (1) and (3),} then it is called a {\em $K$-local discrepancy function.}  

The annotation $\beta$ gives an upper bound on the distance between two neighboring trajectories as a function of their initial states and time. Unlike incremental stability conditions~\cite{angeli2002lyapunov}, the second condition on $\beta$ does not require the trajectories to converge as time goes to infinity, but only as the initial states converge. Obviously, if the function $f$ has a Lipschitz constant $L$, then $\beta(x,x',t) = ||x-x'||e^{Lt}$ meets the above criteria. In~\cite{DMVemsoft2013,HuangFMMK14} other heuristics have been proposed for finding discrepancy functions. As will be clear from the following discussion, the quality of the discrepancy function strongly influences the performance of the simulation-based verification algorithm. \cite{DMVemsoft2013,HuangFMMK14,HuangM:HSCC2014} need user provided discrepancy and simulation engines to give verification of bounded time safety and temporal precedence properties.
In this paper, we will present approaches for computing {\em local discrepancy functions\/} that unburdens the user from finding these annotations.

\subsection{Verification Algorithm}
\label{ssec:simverialgo}

The simulation-based verification algorithm  is shown in Figure~\ref{fig:algo1}.
It takes as input some finite description of the parameters of a safety verification problem, namely,
the function $f$,
the initial set $\Theta$, the unsafe
set $\unsafe$, and the time bound $T$.
It has two main data stuctures:
The first, $\C$ returned by function $\mathit{Partition}$, is a collection of triples $\langle \theta, \delta, \epsilon \rangle$
such that the union of all the $\delta$-balls around the $\theta$'s completely cover the initial set $\Theta$.
The second data structure $\R$ incrementally gets the
bounded-time reachtube from $\Theta$.

Initially, $\C$ has a singleton cover $\langle \theta_0,\delta_0,\epsilon_0\rangle$
such that $\delta_0 = \mathit{dia}(\Theta)$, $\Theta \subseteq B_{\delta_0}(\theta_0)$,
and $\epsilon_0$ is a small constant for simulation precision.

In the {\bf while\/}-loop, this verification algorithm iteratively refines the cover of $\Theta$  and for each $\langle \theta,\delta, \epsilon\rangle$ in $\C$, computes
over-approximations of the reachtube from $B_\delta(\theta)$.
The higher-level structure of the algorithm is familiar:
if the reachtube from $B_\delta(\theta)$ proves to be safe, i.e., disjoint from $\unsafe$, then
the corresponding triple is removed from $\C$ (\lnref{al1:subsafe}).
If part of the reachtube from $B_\delta(\theta)$ overlaps with $\unsafe$, then
the system is declared to be unsafe (\lnref{al1:unsafe}).
Otherwise, a  finer cover of $B_\delta(\theta)$ is created, and the corresponding triples
with finer parameters are added to $\C$.

Here we discuss the reachtubes computed from discrepancy and simulations. For each $\langle \theta,\delta, \epsilon\rangle$ in $\C$,
a $(\theta, \tau, \epsilon, T)$-simulation $\psi$, which is a sequence of $\{(R_i,t_i)\}$, is generated.
Note that $\psi$ contains the trajectory from $\theta$, $\xi(\theta, t), t\in[0, T]$.
Then we bloat each $R_i$ by some factor (\lnref{al1:bloat}) such that the resulting sequence contains the reachtube from $B_\delta(\theta)$. It is shown that this bloated simulation is guaranteed to be an over-approximation of $\reach{}(B_{\delta}(\theta),T)$ and the union of these bloated simulations is an over-approximation of $\reach{}(\Theta,T)$. Therefore, the algorithm is sound. Furthermore, the second property of $\beta$ ensures that the reach set over-approximations become tighter and tighter as we make $\delta$ smaller and smaller. Finally it will return ``SAFE" for robustly safe reachtubes or find a counter example and return ``UNSAFE".
For user defined discrepancy function, the factor is obtained by maximizing $\beta(\theta,\tilde{\theta},t)$
over $\tilde{\theta} \in B_\delta(\theta)$ and $t \in [t_{i-1},t_i]$.

Indeed this is the approach taken in the algorithm presented in~\cite{DMVemsoft2013}.
In this paper, we will analyze in detail the $\findldf$ subroutine which computes a local version of discrepancy function automatically.

\begin{figure}
\centering
\begin{algorithmic}[1]
\STATE{{\bf Input:}{$\Theta,\unsafe,T$}}
\STATE{$\delta \gets \mathit{dia}(\Theta); \epsilon\gets \epsilon_0; \C \gets \emptyset, \R\gets \emptyset$; //$\epsilon_0$ is a small constant\;}
\STATE{$\C   \gets \langle \mathit{Partition}(\Theta,\delta),\delta,\epsilon \rangle$\;}   \lnlabel{al1:partition}
\WHILE{ $\C \neq \emptyset $
}{ 	
	\FOR{$(\theta,\delta,\epsilon)\in \C$}{
		\STATE{$\psi \gets \mathit{Simulate}(\theta, \tau, \epsilon,T)$} \label{al1:simulation}
		\STATE{$\beta \gets$ $\findldf$($\psi,J_f,L_f,\delta,\epsilon$)\;} \lnlabel{al1:beta}
        \STATE{D $\gets$ $\psi \oplus \beta$} \lnlabel{al1:bloat}
		\IF{$D \cap \unsafe =\emptyset$}
			\STATE{$\C\gets \C\backslash\{(\theta,\delta,\epsilon)\} ; \R \gets \R \cup D$ } \lnlabel{al1:subsafe}
		\ELSIF{$\exists k, R_k \subseteq \unsafe$}
        {\RETURN $(\mbox{UNSAFE},\R)$} \lnlabel{al1:unsafe}
		\ELSE{
		\STATE{$\C \gets \C \backslash \{(\theta,\delta,\epsilon)\}$\;}	
		\STATE{$ \C \gets \C \cup \mathit{Partition}(\Theta\cap B_\delta(\theta),(\frac{\delta_1}{2},\dots,\frac{\delta_N}{2}),\frac{\epsilon}{2})$\;}
		}
        \ENDIF
	}
    \ENDFOR
}
\ENDWHILE
\RETURN $(\mbox{SAFE},\R)$\;
\end{algorithmic}
\caption{Verification Algorithm}
\label{fig:algo1}
\end{figure}

The following results from \cite{DMVemsoft2013} state two key properties of the algorithm. Although in \cite{DMVemsoft2013} $\beta$ was defined globally, it is easy to check that the local version still satisfies them.

\begin{theorem}
The Algorithm in Fig.\ref{fig:algo1} is sound, that is, if it returns ``SAFE" then the system is safe; when it returns ``UNSAFE" there exists at least one execution from $\Theta$ that is unsafe. The Algorithm is relatively complete, that is, if the system is robustly safe, the algorithm will terminate and return ``SAFE". If any executions from $\Theta$ is unsafe, it will terminate and return ``UNSAFE".
\end{theorem}

\section{Local discrepancy function}\label{sec:DF}
\label{sec:local}

In this section, we present the analysis of $\findldf$ algorithm.
This algorithm computes a special type of local discrepancy in terms of time-varying exponential functions that  bound from above the distance between two trajectories starting from a compact neighborhood. Roughly speaking, it computes the rate of trajectory convergence or divergence for an interval of time instances.
\begin{definition}
\label{def:localdf}
Consider a compact set $C \subseteq \reals^n$ and a sequence of time points $0 = t_0 < t_1 < t_2 < \ldots < t_k = T$.
For $\forall x_1, x_2 \in C, \forall t\in [0, T]$, a {\em piece-wise exponential discrepancy function\/} 
$\beta: C \times C \times [0,T] \rightarrow \nnreals$ is defined as\\
$\beta(x_1,x_2,t) =$ 
\begin{eqnarray*}
\left\{ \begin{array}{ll}
\|x_1-x_2\|, & \textrm{if}~t = t_0, \\
\beta(x_1,x_2,t_{i-1})e^{b[i](t-t_{i-1})}, & \textrm{if}~t\in (t_{i-1},t_i],\\
\end{array} \right.\\
\end{eqnarray*}
where $b[1], \ldots, b[k]$ are real constants.
\end{definition}
From the definition, we can immediately get that $\beta(x_1,x_2,t) = ||x_1-x_2||e^{b[i](t-t_{i-1}) + \sum_{j=1}^{i-1} b[j] (t_j - t_{j-1})}$, $i = 1, \dots, k$, where $t_{i-1}$ is the largest time point in the sequence before $t$. 
\subsection{ComputeLDF Algorithm}
\label{ssec:findldf}
Figure \ref{alg:ComputeLDF} shows the pseudocode for $\findldf$ used in \lnref{al1:beta} of the verification algorithm. 
$\findldf$ takes as input a parameter $\delta$, an error bound for simulation $\epsilon$, the Lipschitz constant $L_f$, the Jacobian matrix $J_f$ of function $f$, and a $(\theta, \tau, \epsilon, T)$-simulation $\psi=\{(R_i,t_i)\}, i=0,1,\dots,k$.
It computes a 
piece-wise exponential local discrepancy function (LDF) for the compact set $B_\delta(R_0)$ and for the time points $t_0, \dots, t_{k}$. 
and returns it as an array of exponential coefficients $b$.

The algorithm starts with the initial set $B_{\delta}(R_0)$ and with $\Delta=\delta$. In each iteration of the {\bf for\/}-loop it computes exponent $b[i]$ corresponding to the time interval $[t_{i-1},t_i]$.
In the $i^{th}$ iteration, $\Delta$ is updated so that $B_{\Delta}(R_{i-1})$  is an over-approximation of the reachable states from $B_{\delta}(R_0)$ at $ t_{i-1}$ (Lemma \ref{lemma:inductive}).
In~\lnreftwo{alg2:LipDF}{alg2:Lip-bloating}, a set  $S$ is computed by bloating the convex hull $hull(R_{i-1},R_i)$ by a factor of $d =(\Delta+\epsilon) e^{L_f (t_i - t_{i-1})}$.
The set $S$ will later be proved  to be a (coarse) over-approximation of the reachtube from $B_{\Delta}(R_{i-1})$ over the time interval $[t_{i-1},t_{i}]$ (Lemma~\ref{lemma:Lipchitz-reachtube}).
In~\lnsref{alg2:s0}{alg2:upperbound} an upper bound on the maximum eigenvalue of the symmetric part of the Jacobian over the set $S$, is computed as $b[i]$ (Lemma~\ref{lemma:upperbound}). 
Then $\Delta$ is updated as $ (\Delta+\epsilon)  e^{b[i] (t_{i}-t_{i-1})}$ for the next iteration. 
\begin{figure}
\begin{center}
\begin{algorithmic}[1]
\STATE{{\bf Input:} $\psi$,$J_f$,$L_f$,$\delta,\epsilon$}
\STATE{$\Delta \gets \delta$,$b$ $\gets$ zeros(k) \;}
\FOR {i = 1:k}
  \STATE {$\tau \gets t_i - t_{i-1}$}\lnlabel{alg2:start}
	\STATE{$d \gets (\Delta+\epsilon) e^{L_f\tau}$}\lnlabel{alg2:LipDF}
	\STATE{S $\gets$ $hull(R_{i-1},R_{i}) \oplus B_d(0)$} \lnlabel{alg2:Lip-bloating}\;
    \STATE{J $\gets$ $J_f(center(S))$}\; \lnlabel{alg2:s0}
    \STATE{$\lambda$ $\gets \max(eig(J+ J^T)/2)$}\lnlabel{alg2:max-eigen}
    \STATE{error $\gets$ $ \textrm{\chuchu {upper} }_{x\in S} \|(J_f(x)+J_f^T(x))-(J+J^T)\|$} \lnlabel{alg2:maximize}
    \STATE{$b[i]$ $\gets$ $\lambda$ + error/2}\; \lnlabel{alg2:upperbound}
    \STATE{$\Delta$ $\gets$ $(\Delta+\epsilon) e^{b[i] \tau}$}\; \lnlabel{alg2:update delta}
\ENDFOR
\RETURN $b$
\end{algorithmic}
\caption{Algorithm $\findldf$.}
\label{alg:ComputeLDF}
\end{center}
\end{figure}

\subsection{Analysis of ComputeLDF}
\label{sec:analysis}
In this section, we will prove that $\findldf(\psi,J_f,L_f,$ $\delta,\epsilon)$ returns a piece-wise exponential LDF of the system in Equation~(\ref{eq:system}), for the compact neighborhood $B_{\delta}(R_0)$, and the sequence of the time points in the simulation $\psi$.  We establish some lemmas to prove the main theorem.
First, we show in Lemma \ref{lemma:Lipchitz-reachtube} that in the $i^{th}$ iteration of the loop, the computed $S$ is an over-approximation of the set of states that can be reached by the system from $B_{\Delta}(R_{i-1})$ over the time interval $[t_{i-1}, t_i]$.

\begin{lemma}
\label{lemma:Lipchitz-reachtube}
In $i^{th}$ iteration of the loop of $\findldf$, $\reach{}(B_{\Delta}(R_{i-1}),[t_{i-1},t_i]) \subseteq S$.
\end{lemma}
\begin{proof}
Let $\xi(\theta,t)$ denote the actual trajectory from $\theta$, where $\theta$ is the initial state of $\psi$. By Definition~\ref{def:simulation} for $\psi$, it is known that $\theta \in R_0$ and $\forall i = 1,\dots, k, \xi(\theta,t_i) \in R_i$.

For a fixed iteration number $i$, consider state $x = \xi(\theta,t_{i-1}) \in R_{i-1}$ from Definition~\ref{def:simulation}. We know that for any $t \in [t_{i-1},t_i]$, 
$\xi(x,t) \in hull(R_{i-1}, R_{i})$. 
Now consider another state $x' \in B_{\Delta}(R_{i-1})$.
Since $L_f$ is the Lipschitz constant of $f$, using Gronwall's inequality we have that  
$\|\xi(x,t) - \xi(x',t)\| \leq  \|x-x'\|e^{L_f(t - t_{i-1})} $. Since $\|x-x'\| \leq \Delta+\epsilon$, $\|\xi(x,t) - \xi(x',t)\| \leq  (\Delta+\epsilon)e^{L_f(t - t_{i-1})} $. Therefore, $\xi(x',t) \in hull(R_{i-1}, R_{i}) \oplus B_{(\Delta+\epsilon) e^{L_f{(t_i-t_{i-1})}}}(0)=S$. Because $x^\prime$ is arbitrarily selected from $B_{\Delta}(R_{i-1})$, the lemma is proved.
%
\end{proof}

Next, using the generalized mean value theorem (Lemma \ref{lemma:highdimensionmean}), we get that in the $i^{th}$ iteration, the computed $b[i]$ in~\lnref{alg2:upperbound} is the exponential divergence (if positive) or convergence (negative) rate of the distance between any two trajectories starting from $B_{\Delta}(R_{i-1})$ over time $[t_{i-1},t_i]$.

\begin{lemma} 
\label{lemma:highdimensionmean}
For any continuously differentiable vector-valued function $f:\reals^{n} \rightarrow \reals^{n}$, and $x,r \in \reals^{n}$,
\begin{equation}
f(x+r)-f(x) = \left(\int_{0}^{1}{J_f(x+s r)ds} \right) \cdot r,
\end{equation}
where the integral is component-wise.
\end{lemma}

Next, we will use a well-known theorem that gives bounds on eigenvalues of perturbed symmetric matrices, the proof of which uses the Courant-Fischer minimax theorem. 
The complete proofs of Lemma \ref{lemma:highdimensionmean} and Theorem \ref{thm:lamd_ordering} can be found in the appendix.
\begin{theorem} \label{thm:lamd_ordering}
If $A$ and $E$ are $n \times n$ symmetric matrices, then
$$\lambda_n(E) \leq \lambda_k(A+E) - \lambda_k(A) \leq \lambda_1(E),$$
where $\lambda_i(\cdot)$ is the $i^{\textrm{th}}$ largest eigenvalue of a matrix.
\end{theorem}


\begin{corollary}
\label{corollary:sum-of-eigenvalue}
If $A$ and $E$ are $n \times n$ symmetric matrices, then
\begin{equation}\label{perturbation of the matrix}
|\lambda_k(A+E) - \lambda_k(A)| \leq \|E\|.
\end{equation}
\end{corollary}
Since $A$ is symmetric, $\|A\| = \sqrt{\lambda_{\textrm{max}}(A^TA)}= \max({|\lambda(A)|})$.
From Theorem~\ref{thm:lamd_ordering}, we have $|\lambda_k(A+E) - \lambda_k(A)| \leq$ $\max \{|\lambda_n(E)|,$ $|\lambda_1(E)|\}=\|E\|$.
If $E(x)$ is a matrix-valued function: $\reals^n \rightarrow \reals^{n\times n}$ maps a state $x\in \reals^{n}$ to a matrix $E(x)$, and every component of $E(x), e_{ij}(x): \reals^n \rightarrow \reals$ is continuous over some compact closed set $S$, then we can \chuchu {get an upper bound of} $\|E(x)\|$ over $S$ by \chuchu{compute the upper bound of the absolute value of} each term $e_{ij}(x)$, $|e_{ij}(x)|$ over $S$. Let $\textrm{\chuchu {upper}}_{x \in S}(|e_{ij}(x)|)$ be denoted by $ \tilde e_{ij}$, then we know $\chuchu {\forall x \in S},\|E(x)\| \leq$\\ $\sqrt{\sum_{i=1}^{n}\sum_{j=1}^{n}{\tilde e_{ij}^2}}$. \chuchu {Because we assume the system to be Lipschitz continuous, the upper bound of the symmetric part of the Jacobian matrix in \lnref{alg2:maximize} always exists.}
Using Corollary~\ref{corollary:sum-of-eigenvalue}, we next show in Lemma \ref{lemma:upperbound} that  $b[i]$ calculated in \lnref{alg2:upperbound} bounds the eigenvalues of symmetric part of Jacobian matrix over~$S$.

\begin{lemma} 
\label{lemma:upperbound}
In the $i^{th}$ iteration, for $\forall x \in S:
J_f^T(x) + J_f(x)  \preceq 2b[i] I$.
\end{lemma}
\begin{proof}
Let $S$ be the set computed in~\lnref{alg2:Lip-bloating}
and $J$ be the Jacobian evaluated at the center $s_0$ of $S$.
Consider any point $x \in S$. 
We define the perturbation matrix $E(x) \equiv J_f^T(x) + J_f(x) - (J^T +J)$.
Since $J_f^T(x)+J_f(x)$ and $J^T+J$ are symmetric matrices,  Corollary~\ref{corollary:sum-of-eigenvalue} implies that 
$\lambda_{max}(J_f^T(x) + J_f(x)) - \lambda_{max}(J^T +J) \leq ||E(x)||$. 
The $error$ term computed in \lnref{alg2:maximize} is the upperbound on $||E(x)||$. Therefore,
$\lambda_{max}(J_f^T(x) + J_f(x)) \leq \lambda_{max}(J^T +J) + error$.
In \lnref{alg2:upperbound} set $b[i]$  equals to $\lambda_{max}((J^T +J)/2) + error/2$. Thus, $\lambda_{max}(J_f^T(x) + J_f(x)) \leq 2 b[i]$, which immediately indicates that $\forall x \in S: J_f^T(x) + J_f(x)  \preceq 2b[i] I$.
\end{proof}

By Lemma \ref{lemma:highdimensionmean} and Lemma \ref{lemma:upperbound}, we can prove as in Lemma \ref{lemma:DF onestep} that  $b[i] $ calculated in \lnref{alg2:upperbound} is the exponential rate of divergence or convergence of two trajectories starting from $B_{\Delta}(R_{i-1})$ over the interval $[t_{i-1},t_i]$.

\begin{lemma}
\label{lemma:DF onestep}
In the $i^{th}$ iteration, for any two states $x_1, x_2 \in B_{\Delta}(R_{i-1})$ at time $t_{i-1}$, 
and any time $t \in [t_{i-1},t_{i}]$, 
$\| \xi(x_1 ,t) - \xi(x_2,t) \| \leq \|x_1-x_2\|e^{b[i](t-t_{i-1})}$.
\end{lemma}

\begin{proof}
Let us fix the iteration $i$ and two states $x_1, x_2 \in B_{\Delta}(R_{i-1})$. From Lemma \ref{lemma:Lipchitz-reachtube} it's can be seen that for any $t\in[t_{i-1},t_{i}]$, $\xi(x_1,t)\in S, \xi(x_2,t) \in S$.
Define $y(t)\equiv \xi(x_2,t)-\xi(x_{1},t)$. For a fixed time $t$, from Lemma~\ref{lemma:highdimensionmean} we have
\begin{eqnarray} \label{eqn:ODE-delta}
\dot y(t) &=& \dot \xi(x_2,t)-\dot \xi(x_{1},t) = f(\xi(x_2,t))-f(\xi(x_1,t)) \nonumber \\
&=& \left( \int_{0}^{1}{J_f(\xi(x_1,t)+ sy(t))ds}\right) y(t). 
\end{eqnarray}
Since $S$ is the Minkowski sum of two convex sets $hull(R_{i-1},R_i)$ and $B_{\Delta e^{L_f{(t_i-t_{i-1})}}}(0)$, it is also convex. Recall that  $\xi(x_1,t),$ $\xi(x_2,t) \in S$, and  
for any $s \in [0,1]$, $\xi(x_1,t)+s y(t) \subseteq S$.

Differentiating $\|y(t)\|^2$, we have

\begin{eqnarray} \label{eqn:odeof_y}
\frac{d\|y(t)\|^2}{d t}  &=& \dot y^T(t) y(t) + y^T(t)\dot y(t) \nonumber \\
&=&  y^T(t) \left( \int_{0}^{1}{J_f^T(\xi(x_1,t)+sy(t))ds}\right) y(t)  \nonumber \\
&+& y^T(t) \left( \int_{0}^{1}{J_f(\xi(x_1,t)+s y(t))ds}\right) y(t). \nonumber \\
\end{eqnarray}

Using Lemma~\ref{lemma:upperbound}, we know 
$$
\forall x \in S, ~~J_f^T(x)+J_f(x) \preceq 2 b[i]I.
$$
Thus, we can bound \eqref{eqn:odeof_y}
\begin{eqnarray} \label{eqn:odeof_y2}
\frac{d\|y(t)\|^2}{d t} &\leq& y^T(t)\left( \int_{0}^{1}{(2b[i]I) ds}\right)y(t) \nonumber \\
&=& 2b[i] y^T(t)  y (t) \nonumber \\
&=& 2b[i] \|y(t)\|^2.
\end{eqnarray}

Integrating both sides 
over $t_{i-1}$ to any $t \in [t_{i-1},t_{i}]$, we have
\begin{equation*}\label{convergent rate 2}
\begin{split}
&\ln (\|y(t)\|^2)-\ln(\|y(t_{i-1})\|^2) \leq 2b[i](t-t_{i-1}) \\
\Rightarrow &\|y(t)\|^2 \leq \|y(t_{i-1})\|^2 e^{2b[i](t-t_{i-1})} \\
\Rightarrow & \| \xi(x_1 ,t) - \xi(x_2,t) \| \leq \|x_1-x_2\|e^{b[i](t-t_{i-1})}.
\end{split}
\end{equation*}
\end{proof}

Up to this point all the lemmas were statements about a single iteration of the {\bf for \/}-loop, next we show that in $i^{th}$ iteration of the loop, $B_{\Delta}(R_i)$ used in Lemma \ref{lemma:Lipchitz-reachtube} and \ref{lemma:DF onestep} is the reach set from $B_{\delta}(R_0)$ at time $t_i$.
\begin{lemma} \label{lemma:inductive}
For $\forall i = 1,\dots, k$, $\Reach{}(B_{\delta}(R_0),[t_i,t_i]) \subseteq B_{\Delta_i}(R_i)$, and $\Reach{}(B_{\Delta_{i-1}}(R_{i-1}),[t_{i-1},t_{i}]) \subseteq hull(R_{i-1},R_{i}) \oplus B_{\Delta'_i}(0)$, where $\Delta_i$ is $\Delta$ after \lnref{alg2:update delta} is executed in the $i^{th}$ iteration, and $\Delta'_i = \max \{\Delta_i, \Delta_{i-1}+\epsilon\}$
\end{lemma}
\begin{proof}

In this proof,
let $\xi(\theta,\cdot)$ denote the trajectory from $\theta$. From the Definition~\ref{def:simulation} for $\psi$, we know that $\theta \in R_0$ and $\forall i = 1,\dots, k, \xi(\theta,t_i) \in R_i$. Let $S_i$ denote $S$ after \lnref{alg2:Lip-bloating} is executed in the $i^{th}$ iteration.
The lemma is proved by induction on $i$.
Note that the initial set is $B_{\delta}(R_0)$, and before the {\bf for\/}-loop, $\Delta_0$ is set as $\delta$. 

When $i=1$, we already have $\Reach{}(B_{\delta}(R_0),[t_0,t_0]) = B_{\delta}(R_0) = B_{\Delta_0}(R_0)$.

Lemma~\ref{lemma:Lipchitz-reachtube} indicates that
$\forall t \in [t_0,t_1], \Reach{}(B_{\Delta_0}(R_0),$ $[t_0,t_1]) \subseteq S$. 
And consider state $x =\theta \in R_0$, we also know $\xi(x,t) \in hull(R_0, R_1)$ and $\xi(x,t_1) \in R_1$. 
From Lemma \ref{lemma:DF onestep}, it follows that for $\forall x' \in B_{\Delta_0}(R_0),\forall t \in [t_0,t_{1}] $,
\begin{equation*} \label{eqn:t0-t1}
\begin{aligned}
\|\xi(x ,t) - \xi(x',t) \| \leq\|x-x'\|e^{b[1](t-t_0)}.
\end{aligned}
\end{equation*} 
And at \lnref{alg2:update delta}, $\Delta_1 \gets (\Delta_0+\epsilon)e^{b[1](t_1-t_0)}$. Since $b[1]$ could be positive or negative, $\max_{t \in [t_0,t_1]} {\|x-x'\|e^{b[1](t-t_0)}}=$\\ $\max \{\Delta_1, \Delta_0 + \epsilon\}$.
Therefore, $$\Reach{}(B_{\delta}(R_0),[t_0,t_1]) \subseteq hull(R_0, R_1) \oplus B_{\max\{\Delta_1,\Delta_0+\epsilon\}}(0),$$ and at time $t_1$, $\xi(x',t_1)$ is at most $\Delta_1$ distance to $\xi(x,t_1)\in R_1$, 
so $\Reach{}(B_{\delta}(R_0),[t_1,t_1]) = \Reach{}(B_{\Delta_0}(R_0),[t_1,t_1]) \subseteq B_{\Delta_1}(R_1)$.

Assuming that  the lemma holds for $i=m-1$, we have $\Reach{}(B_{\delta}(R_0),[t_{m-1},t_{m-1}]) \subseteq B_{\Delta_{m-1}}(R_{m-1})$. Next we prove the lemma holds for $i=m$ as well. Consider state $x = \xi(\theta,t_{m-1}) \in R_{m-1}$, $\forall t \in [t_{m-1},t_m]$, by definition it follows that $\xi(x,t) \in hull(R_{m-1}, R_{m})$ and $\xi(x,t_m) \in R_m$.
$\forall x' \in B_{\Delta_{m-1}}(R_{m-1}),\forall t \in [t_{m-1},t_{m}] $, from Lemma \ref{lemma:DF onestep}
\begin{equation*} \label{eqn:t0-ti}
\begin{aligned}
\|\xi(x ,t) - \xi(x',t) \| \leq\|x-x'\|e^{b[m](t-t_{m-1})}.
\end{aligned}
\end{equation*} 
Note at \lnref{alg2:update delta}, $\Delta_m \gets (\Delta_{m-1}+\epsilon)e^{b[m](t_{m}-t_{m-1})}$.
Therefore, 
$\Reach{}(B_{\Delta_{m-1}}(R_{m-1}),[t_{m-1},t_m]) \subseteq hull(R_{m-1},$ $R_m) \oplus B_{\max\{\Delta_m,\Delta_{m-1}+\epsilon\}}(0)$.
And at time $t_m$, $\xi(x',t_m)$ is at most $\Delta_m$ distance to $\xi(x,t_m)\in R_m$.
Hence, $\Reach{}(B_{\Delta_{m-1}}(R_{m-1}),$ $[t_{m},t_m])$ $\subseteq B_{\Delta_m}(R_m)$.
Recall that $\Reach{}(B_{\delta}(R_0),[t_{m-1},t_{m-1}])$ $\subseteq B_{\Delta_{m-1}}(R_{m-1})$, thus $\Reach{}(B_{\delta}(R_0),[t_m,t_m]) \subseteq B_{\Delta_m}(R_m)$.
\end{proof}

$\cup_{i=1}^k\{hull(R_{i-1},R_{i}) \oplus B_{\Delta'_i}(0)\}$  contains the $(B_\delta(R_0), \tau, T)$-reachtube of the system. \lnref{al1:bloat} of the algorithm in Figure \ref{fig:algo1} is computed in this way.
Now we are ready to prove the main theorem. 
\begin{theorem} \label{theorem: main}
The items in array $b$ computed by \\ $\findldf$ are the coefficients of a 
$B_{\delta}(R_0)$-local piece-wise exponential discrepancy function (Definition \ref{def:localdf}).
\end{theorem}
\begin{proof}

First of all consider any time $t \in [t_0,t_1]$ and any two states: $x_1,x_2 \in B_{\delta}(R_0)$. By Lemma \ref{lemma:DF onestep}, $\|\xi(x_1,t)-\xi(x_2,t)\| \leq \|x_1-x_2\|e^{b[1](t-t_0)}$.
Then consider $t\in [t_1, t_2]$. By Lemma \ref{lemma:inductive} we know at time $t_{1}$, $\xi(x_1,t_{1})$ and $\xi(x_2,t_{1})$ are all contained in $B_{\Delta_{1}}(R_{1})$, so we can use Lemma \ref{lemma:DF onestep} such that for any time $t \in [t_1,t_2]$, 
$\|\xi(x_1,t)-\xi(x_2,t)\| \leq \|\xi(x_1,t_1)-\xi(x_2,t_1)\|e^{b[2](t-t_1)} \leq \|x_1-x_2\|e^{b[2](t-t_1)+b[1](t_1-t_0)}$.

The procedure above can be performed iteratively as follows. For any time $t \in [t_{i-1},t_{i}]$, by lemma \ref{lemma:inductive} we know at time $t_{i-1}$, $\xi(x_1,t_{i-1})$ and $\xi(x_2,t_{i-1})$ are all contained in $B_{\Delta_{i-1}}(R_{i-1})$. By Lemma \ref{lemma:DF onestep} it follows that 
\begin{eqnarray*}
\|\xi(x_1,t)-\xi(x_2,t)\|&\leq& \|\xi(x_1,t_{i-1})-\xi(x_2,t_{i-1})\|  e^{b[i](t-t_{i-1})}\\
&\leq&\|x_1-x_2\|  e^{b[i](t-t_{i-1})+\sum_{j=1}^{i-1}b[j](t_j-t_{j-1})}.
\end{eqnarray*}
Next we will prove that \\
$\beta(x_1,x_2,t) \equiv \|x_1-x_2\|  e^{b[i](t-t_{i-1})+\sum_{j=1}^{i-1}b[j](t_j-t_{j-1})} $ is a valid LDF.

In \lnsref{alg2:s0}{alg2:upperbound}, because $J$ is a real matrix, the maximum eigenvalue $\lambda$ of $(J^T+J)/2$ is bounded. Assume that each component of $E(x)= J_f^T(x)+J_f(x)-J^T-J$ is continuous over the closed set $S$, \chuchu {then we can find the upper bound of $\|E(x)\|$}, so the ``error" term is also bounded. Therefore, each $b[i]$ is bounded.
So $\forall t \in [t_{i-1},t_{i}]$, $i= 1,\dots,k$, $\exists N < \infty$, such that  $e^{b[i](t-t_{i-1})+\sum_{j=1}^{i-1}b[j](t_j-t_{j-1})}$ is bounded by $N$ from the above. 

As $x_1 \rightarrow x_2$, obviously, 
$$\|x_1-x_2\|  e^{b[i](t-t_{i-1})+\sum_{j=1}^{i-1}b[j](t_j-t_{j-1})} \rightarrow 0.$$
And for any $\epsilon > 0$, $\exists \delta = \epsilon/N > 0$, such that $\forall x_1,x_2 \in B_{\delta}(R_0)$ and $\|x_1-x_2\| < \delta$, it follows
$$\|x_1-x_2\|  e^{b[i](t-t_{i-1})+\sum_{j=1}^{i-1}b[j](t_j-t_{j-1})} < \epsilon/N\cdot N = \epsilon.$$
 So $\beta(x_1,x_2,t) = \|x_1-x_2\|  e^{b[i](t-t_{i-1})+\sum_{j=1}^{i-1}b[j](t_j-t_{j-1})} $ is a $B_{\delta}(R_0)$-local piece-wise discrepancy function and the array $b$ contains the corresponding coefficients.
\end{proof}

\subsection{Coordinate transformation}
\label{sec:cood}
In this section, we will discuss the issue that the upper bound of the symmetric part of the Jacobian computed in~\lnsref{alg2:s0}{alg2:upperbound} may introduce loss in precision.
We propose a a strategy to reduce this loss by first performing a coordinate transformation.
Consider a simple linear system:  
\begin{equation}\label{linear transformation example}
\dot x = \left[ \begin{array}{cc}
	0 & 3 \\
	-1& 0
\end{array} 
\right]x,
\end{equation}
which has eigenvalues $\pm \sqrt {3}i$ and thus its trajectories oscillate. 
The symmetric part the of the Jacobian is $\left[ \begin{array}{cc}
0 & 1 \\
1& 0
\end{array}
\right]$ 
with eigenvalues $\pm 1$, which gives the exponentially growing discrepancy with $b = 1$.
In what follows, we will see that a tighter bound can be obtained by first taking linear transformation of $x$. 
The following is a coordinate transformed version of Lemma~\ref{lemma:DF onestep}.
The coordinate transformation matrix $P$ can be any $n \times n$ real invertible matrix, and the condition number of $P$ is $\|P\|\|P^{-1}\|$.

\begin{lemma}\label{lemma:coordinate transfer}
In $i^{th}$ iteration of the loop, for any $x_1, x_2 \in B_{\Delta}(R_{i-1})$, 
and any $t \in [t_{i-1},t_{i}]$, 
\[
\| \xi(x_1 ,t) - \xi(x_2,t) \| \leq K\|x_1-x_2\|e^{\tilde \lambda_\textrm{max}(S)(t-t_{i-1})},
\]
where $\tilde \lambda_{max}(S)$ is the upper bound of 
$\frac{1}{2}(\widetilde{J_f}^T(x)+\widetilde{J_f}(x))$ over the set $S$, 
$\widetilde{J_f}(x) = P J_f(x) P^{-1}$, and 
$K$ is the condition number of $P$.
\end{lemma}
 
\begin{proof}
Let $z(t) = \xi (x_{2},t) - \xi(x_{1},t)$ and $y(t) = P z(t)$
From \eqref{eqn:ODE-delta} get:
\begin{eqnarray*}
 \dot y(t) &=& P \dot z(t)  \\
&=& P\left( \int_{0}^{1}{J_f(\xi(x_1,t)+sz(t))ds}\right) z(t) \\
&=& P\left( \int_{0}^{1}{J_f(\xi(x_1,t)+sz(t))ds}\right) P^{-1} y(t) \\
&=& \left( \int_{0}^{1}{\widetilde{J_f}(\xi(x_1,t)+sz(t))ds}\right) y(t).
\end{eqnarray*}
Since for all $x \in S, \widetilde{J_f}^T(x)+\widetilde{J_f}(x) \preceq  \tilde \lambda_{\textrm{max}}(S)I$ and $\forall s \in [0,1], \xi(x_1,t)+sz(t) \subseteq S$, we have
\begin{eqnarray*}
\frac{d\|y(t)\|^2}{dt} \leq 2 \tilde \lambda_{\textrm{max}}(S) \|y(t)\|^2,
\end{eqnarray*}
which leads to:
$\forall t \in [t_{i-1},t_i]$
\begin{equation} \label{eqn:delta y}
\|y(t)\| \leq \|y(t_{i-1})\|e^{\tilde \lambda_{\textrm{max}}(S)(t-t_{i-1})}.
\end{equation}

Substituting \eqref{eqn:delta y} in $z(t) = P^{-1} y(t)$:
\begin{eqnarray}
\| z(t) \| & \leq &\| P^{-1} \| \| y(t) \| \nonumber \\ 
& \leq & \| P^{-1} \| \| y(t_{i-1}) \|e^{ \tilde \lambda_{max}(S) (t-t_{i-1})} \nonumber \\
& \leq & \| P^{-1} \| \|P \| \|z(t_{i-1}) \|e^{\tilde \lambda_{max} (S) (t-t_{i-1})} \nonumber \\
& = & cond(P) \|z(t_{i-1}) \|e^{ \tilde \lambda_{max} (S) (t-t_{i-1})}.
\end{eqnarray}
\end{proof}

This shows that the distance can be bounded in the same way for the transformed system with a (possibly much smaller ) $\tilde \lambda_{max}(S)$ but with an additional multiplicative cost of $cond(P)$.

\chuchu {To choose the coordinate transformation matrix, one approach that produces good empirical results is making the Jacobian matrix at the center point a real Jordan form. Let $S$ be the set computed in~\lnref{alg2:Lip-bloating}
and $J$ is the Jacobian evaluated at the center $s_0$ of $S$.} Let $\tilde J = PJP^{-1}$ the real Jordan form which looks like:

\[ \left[ \begin{array}{cc|c|cc}
\lambda _1 & \epsilon & 0 & 0 & 0\\
0 & \lambda _1 & 0 & 0 & 0\\
\hline
0 & 0 & \lambda _2 & 0 & 0\\
\hline
0 & 0 & 0 & \lambda _3 & c\\
0 & 0& 0 & -c & \lambda _3
\end{array} \right]\]
where $2\epsilon < \lambda _1$ and $\lambda _1, \lambda _2 , \lambda _3\pm ci$ are the eigenvalues of $J$. There could be several more blocks like 
$\left[ \begin{array}{cc}
\lambda _1 & \epsilon \\
0 & \lambda _1 \\
\end{array} \right]$
,$\lambda_2$ and 
$\left[ \begin{array}{cc}
\lambda _3 & c \\
-c & \lambda _3 \\
\end{array} \right]$ in general.
We use the matrix $P$ as the coordinate transformation matrix for $J_f(x)$.
In this approach the eigenvalues of $\frac{1}{2}(\tilde J+ \tilde J^T)$ are $\lambda _1 + \frac{\epsilon}{2}, \lambda _2 , \lambda _3$, which preserve the original eigenvalues to some extent. 
\chuchu {Contraction matrix \cite{slotine98} introduces more general coordinate transformation. However, there are no general methods to compute it for nonlinear systems. Choosing a constant matrix as transformation is an implementable approach, and Lemma \ref{lemma:coordinate transfer} applies to any invertible matrix.}

\chuchu{Combining the linear transformation, we can get the algorithm in Figure \ref{alg:ComputeLDFwithCT}.}
\begin{figure}
\begin{center}
\begin{algorithmic}[1]
\STATE{{\bf Input:} $\psi$,$J_f$,$L_f$,$\delta,\epsilon$,step}
\STATE{$\Delta \gets \delta$,$b$ $\gets$ zeros(k), $K \gets$ zeros(ceil(k/step)) \;}
\FOR{j = 1:step:k}
\STATE{$[V,D]$ = JordanDecom(average({$R_{j}$,...,$R_{j+step-1}$}))\;}
\STATE{$K(ceil(k/step))) \gets cond(V)$}
\FOR {i = j:j+step-1}
  \STATE {$\tau \gets t_i - t_{i-1}$}\lnlabel{alg2:start}
	\STATE{$d \gets (\Delta+\epsilon) e^{L_f\tau}$}\lnlabel{alg2:LipDF}
	\STATE{S $\gets$ $hull(R_{i-1},R_{i}) \oplus B_d(0)$} \lnlabel{alg2:Lip-bloating}\;
    \STATE{J $\gets$ $V$$J_f(center(S))$}$V^{-1}$\; \lnlabel{alg2:s0}
    \STATE{$\lambda$ $\gets \max(eig(J+ J^T)/2)$}\lnlabel{alg2:max-eigen}
    \STATE{error $\gets$ $ \textrm{{upper} }_{x\in S} \|V((J_f(x)+J_f^T(x))-(J+J^T))*V^{-1}\|$} \lnlabel{alg2:maximize}
    \STATE{$b[i]$ $\gets$ $\lambda$ + error/2}\; \lnlabel{alg2:upperbound}
    \STATE{$\Delta$ $\gets$ $(\Delta+\epsilon) e^{b[i] \tau}$}\; \lnlabel{alg2:update delta}
\ENDFOR
\STATE{$\Delta$ $\gets K(ceil(k/step))) \Delta$}\; 
\ENDFOR
\RETURN $b,K$
\end{algorithmic}
\caption{Algorithm $\findldf$ to coordinate transformation.}
\label{alg:ComputeLDFwithCT}
\end{center}
\end{figure}

In the previous example \eqref{linear transformation example},
the Jacobian matrix is constant, and the discrepancy function without coordinate transformation is:
\[
\|\xi (x_{1},t)-\xi (x_{2},t)\| \leq \|x_{1}-x_{2}\|e^{t-t_1}.
\]
If we use $P = \left[ \begin{array}{cc}
1 & 3 \\
-\sqrt{3} & \sqrt{3} \\
\end{array} \right]$ as the coordinate transformation matrix, $\tilde J = PJP^{-1} = 
\left[ \begin{array}{cc}
0 & \sqrt{3} \\
-\sqrt{3} & 0 \\
\end{array} \right],$
and the discrepancy function with coordinate transformation is 
\[
\|\xi (x_{1},t)-\xi (x_{2},t)\| \leq \sqrt{3}\|x_{1}-x_{2}\|.
\]
In practice, the coordinate transformation can be made for longer time interval $[t_{i-k},t_i]$, where $k>2$, to reduce the multiplicative error term $\prod cond(P[i])$. 
\section{Local Input-State Discrepancy} \label{sec:ISDF}
\label{sec:lisd}
Large and complex models of dynamical system are created by composing smaller modules or subsystems. Consider a dynamical system $A$ consisting of several interacting subsystems $A_1,\dots,A_N$, that is, the input signals of a subsystem $A_i$ are driven by the outputs (or states) of some another component $A_j$. Let's say that each $A_i$ is $n$-dimensional which makes $A$ $nN$-dimensional. One way of achieving scalable verification of $A$ is to exploit this compositional structure and somehow analyze the component $A_i$'s to infer properties of $A$.

In~\cite{HuangM:HSCC2014}, the notion of input-to-state (IS) discrepancy was introduced to address the problem of finding annotations for large models. It is shown that if we can find input-to-state (IS) discrepancy functions for the individual component $A_i$, then we can construct a reduced $N$-dimensional model $M$ such that the executions of $M$ serve as the discrepancy of the overall system. Thus, from IS-discrepancy for the smaller $A_i$ models and simulations of the  $N$-dimensional system $M$, we are able to verify $A$. This has the beneficial side-effect that if the $A_i$'s are rewired in a new topology, then only the reduced model changes~\cite{HuangFMMK14}. However,\cite{HuangM:HSCC2014} still assumes that the user provides the IS-discrepancy for the smaller modules. In this section, we will show the approach used in previous section can be used to get IS discrepancy function for Lipschitz continuous nonlinear subsystems $A_i$. Furthermore, it gives an over-approximation of the reachsets with nondeterministic bounded inputs.

\subsection{Defining Local IS Discrepancy}
\label{sec:LISD}
Consider a dynamical system with inputs: 
\begin{eqnarray}\label{eqn:system-input}
\dot x = f(x,u) \label{eq:sysinputs}
\end{eqnarray}
where $f: \reals^n \times \reals^p \rightarrow \reals^n$ 
is Lipschitz continuous. 
For a given input signal which is a integrable function $\upsilon:[0,\infty) \rightarrow \reals^p$, and an initial state $x_0 \in \reals^n$, a solution (or trajectory) of the system is a function $\xi: \reals^n \times \nnreals \rightarrow \reals^n$ such that $\xi(x_0,0) = x_0$ and for any time $t\geq0$, $\dot \xi(x,t) = f(\xi(x,t),\upsilon(t))$. 

First, we give the original definition of IS discrepancy function for the system in \eqref{eqn:system-input}.
Here $\U$ is the set $\{u | u:[0,\infty) \rightarrow \reals^p\}$ of all input signals.
\begin{definition}
A pair of  uniformly continuous functions $\beta: \nnreals \times \nnreals \rightarrow \nnreals$ and 
$\gamma: \nnreals \rightarrow \nnreals$ is called $C$-local input-to-state discrepancy if
\begin{enumerate}[(1)]
\item $\beta$ is of class $\K$ with respect to its first argument and $\gamma$ is also of class $\K$,
\item for any pair of initial states $x,x' \in$ $C$, any pair of input signals $u,u' \in \U$, and $t \in \nnreals$:\begin{equation} \label{ISDF}
\|\xi(x,t)-\xi(x',t)\| \leq \beta(\|x-x'\|,t) + \int_{0}^{t}{\gamma(\|u(s)-u'(s)\|)ds}.
\end{equation}
\end{enumerate}
\end{definition}




For a bounded, compact set $\I \subseteq \reals^p$. A family of bounded time input signals over $\I$ is the set $\U(\I)=\{u| u: [0,T) \rightarrow \I \}$ of integrable functions.
We denote $\Reach{}(K, \U(\I), [t_1,t_2])$ as the reachable states of the system from compact set $K$ with input set $\U(\I)$ over $[t_1,t_2]$.
Next, we introduce an inductive definition of IS discrepancy for inputs over compact neighborhoods.
\begin{definition}
Consider compact sets $K \in \reals^n,\I \in \reals^p$ and a sequence of time points $0 = t_0 < t_1 < t_2 < \ldots < t_k = T$.
For any pair of initial states $x_1, x_2 \in K$, any pair of input signals $u_1,u_2 \in \U(\I)$, the $(K,\U(\I))$-local IS discrepancy function $\alpha: K^2 \times \U(\I)^2  \times \nnreals \rightarrow \nnreals$ is defined as:

$\alpha(x_1,x_2,u_1,u_2,t) =$ 
\begin{eqnarray*}
\left\{ \begin{array}{l}
\|x_1-x_2\|,~~ \textrm{if}~t = t_0, \\
\alpha(x_1,x_2,u_1,u_2,t_{i-1})e^{a[i](t-t_{i-1})}\\
+M[i]e^{a[i](t-t_{i-1})}\int_{t_{i-1}}^{t}{\|u_1(\tau)-u_2(\tau)\|d \tau}, \textrm{if}~t\in (t_{i-1},t_i],\\
\end{array} \right. \\
\end{eqnarray*}

where $a[1], \dots, a[k], M[1], \dots, M[k]$ are real constants. 
\end{definition}

\subsection{Algorithm for Local IS Discrepancy}
The approach to find $(K,\U(\I))$-local IS discrepancy function is similar to $\findldf$ algorithm, which also uses a {\bf for \/}-loop to compute the coefficients $a[i]$ and $M[i]$. The only changes are 1) in \lnref{alg2:Lip-bloating} $S$ should be computed as in Lemma \ref{lemma:islip}, 2) in \lnref{alg2:update delta} $\Delta$ should be updated as in Lemma \ref{thm:ISDF}. Next we illustrate this process in more detail.
First, we use Lipschitz constant to get a coarse over-approximation of $\Reach{}(K,\U(\I), [t_{i-1},t_{i}])$ parallel to Lemma \ref{lemma:Lipchitz-reachtube}. Let $l = dia(\I)$.
\begin{lemma} \label{lemma:islip}
In $i^{th}$ iteration of the {\bf for \/}-loop, $\Reach{}(B_{\Delta}(R_{i-1}),$ $\U(\I), [t_{i-1},t_{i}]) \subseteq S$, where $S = hull(R_{i-1},R_{i}) \oplus B_{\Delta'}(R_i)$ and $\Delta' = (\Delta+\epsilon)(e^{L_f\tau_i})+ lL_fe^{L_f\tau_i}\tau_i$, $\tau_i = t_{i}-t_{i-1}$.
\end{lemma}
Two trajectories starting from $x_1,x_2 \in \reals^n$ at $t_{i-1}$, with $u_1,u_2 \in \U(\I)$ as inputs respectively, their distance at time $t$,
$\|\xi(x_1,t)-\xi(x_2,t)\| \leq \|x_1-x_2\|(e^{L_f(t-t_{i-1})})+$ $L_fe^{L_f(t-t_{i-1})} \cdot $ $ \int_{t_{i-1}}^{t}{\|u_1(\tau)-u_2(\tau)\|d\tau}$. The lemma directly follows this inequality.

Next we give a one step IS discrepancy function in Lemma \ref{thm:ISDF}.
Before proving it, we need another generalized form of mean value theorem:
\begin{lemma}\label{lemma: high dimension with input mean value function}
For any continuous and differentiable function $f:\reals^n \times \reals^p \rightarrow \reals^n$, $f(x+r,u+w) - f(x,u) =$ \\$\left(\int_{0}^{1}{J_x(x+sr,u+w)ds}\right)r + $ $\left( \int_{0}^{1}{J_u(x,u+\tau w)d \tau} \right) w$,
where $J_x = \frac{\partial f(x,u)}{\partial x}$ and $J_u = \frac{\partial f(x,u)}{\partial u}$
are the Jacobian matrices of $f$ with respect to $x$ and $u$.
\end{lemma}
\begin{proof}
The lemma follows 
by writing $f(x+r,u+w) - f(x,u) = f(x+r,u+w) - f(x,u+w) + f(x,u+w)-f(x,u)$
and then invoking Lemma~\ref{lemma:highdimensionmean}.
\end{proof}

\begin{lemma}\label{thm:ISDF}
Consider the $i^{th}$ iteration of the loop for a dynamic system \eqref{eqn:system-input}.
Let $x,x' \in B_{\Delta}(R_{i-1})$, and $\xi(x,t)$, $\xi(x',t)$ be the trajectories  starting from $x$ and $x'$ with input $u_1(t),u_2(t) \in \U(\I)$ respectively, where $t\in [t_{i-1},t_i]$.
Then,
\begin{eqnarray}
&&\|\xi(x,t) - \xi(x',t)\|  \leq \|x-x'\| e^{a(t-t_{i-1})}\nonumber \\
&+& Me^{a(t-t_{i-1})}\int_{t_{i-1}}^{t}{\|u_1(\tau) - u_2(\tau)\| d \tau},
\end{eqnarray}
where $a = \lambda_{max}(S)+\frac{1}{2}$,
$\lambda_{max}(S)$ is the upperbound of the eigenvalues of the symmetric part of $J_x$ over $S$,
and $M = \max\limits_{u \in \U(\I)}{\left( \| J_u(\xi(x,t),u)\| \right)}$.
\end{lemma}
\begin{proof}
 let $y(t) = \xi(x',t) - \xi(x,t)$ and $v(t) = u_2(t) - u_1(t)$.
For a fixed time $t$, using Lemma \ref{lemma: high dimension with input mean value function}
\begin{eqnarray*}
\dot y(t) &=& \dot \xi(x',t) - \dot \xi(x,t) \nonumber \\
&=& f(\xi(x',t),u_2(t))-f(\xi(x,t),u_1(t))  \nonumber \\
& = & \left( \int_{0}^{1}{J_x(\xi(x,t)+s y(t),u_2(t))ds} \right) y(t) \nonumber \\
&+& \left( \int_{0}^{1}{J_u(\xi(x,t),u_1(t)+\tau v(t))d \tau} \right)v(t) .
\end{eqnarray*}
We write $J_x(\xi(x,t)+s y(t),u_2(t))$ as $J_x$ and $J_u(\xi(x,t),u_1(t)+\tau v(t))$ as $J_u$.
Then the differentiating $\|y(t)\|^2$ with respect to $t$:
\begin{equation} \label{ineqn:ode}
\begin{aligned}
&\frac{d}{dt}\|y(t)\|^2 = y^T(t) \left( \int_{0}^{1}{(J_x^T+J_x) ds}\right) y(t) \\
&+ v^T(t) \left( \int_{0}^{1}{J_u^Td \tau} \right) y(t)+ y^T(t) \left( \int_{0}^{1}{J_u d \tau} \right) v(t)   \\
&\leq y^T(t) \left( \int_{0}^{1}{J_x^T+J_x ds} \right) y(t) + y^T(t) y(t)  \\
&+ \left(\left( \int_{0}^{1}{J_u d \tau} \right) v(t)\right)^T \left(\left( \int_{0}^{1}{J_u d \tau} \right) v(t)\right). \\ 
\end{aligned}
\end{equation}

Recall that $\lambda_{max}(S)$ is the upperbound of the eigenvalues of the symmetric part of $J_x$ over $S$, so $J^T_x+J_x \preceq 2 \lambda_{max}(S)I$.
Therefore, \eqref{ineqn:ode} becomes:

$$\frac{d}{dt}\|y(t)\|^2 \leq (2\lambda_{max}(S)+1) \|y(t)\|^2 + \|\left( \int_{0}^{1}{J_u d \tau} \right) v(t)\|^2.$$

Let $2a = 2\lambda_{max}(S)+1$, $M = \max\limits_{u \in \U(\I)}{\left( \| J_u(\xi(x,t),u)\| \right)},$
then equation \eqref{ineqn:ode} becomes
\begin{equation}
\frac{d}{dt}\|y(t)\|^2 \leq 2a \|y(t)\|^2 + M \|v(t)\|^2.
\end{equation}
Integrating each side from $t_{i-1}$ to $t$ where $t<t_i$, we have:
\begin{eqnarray}
\|y(t)\|^2 \leq e^{2a(t-t_{i-1})}\left(\| y(t_{i-1})\|^2 + \int_{t_{i-1}}^{t}{M\|v(\tau)\|^2 d \tau} \right).
\end{eqnarray}
It follows that,
\[
\|y(t)\| \leq e^{a(t-t_{i-1})}\| y(t_{i-1})\| + Me^{a(t-t_{i-1})}\int_{t_{i-1}}^{t}{\|v(\tau)\| d \tau}.
\]
\end{proof}

Using Lemma \ref{thm:ISDF} to get the coefficients $a[i]$ and $M[i]$ in each time interval $[t_{i-1},t_i], i=1\dots,k$, we will have:
\begin{theorem}\label{thm:ISmain}
The items in array $a$ and $M$ are a coefficients of the $(K,\U(\I))$-local IS discrepancy function for the system \eqref{eqn:system-input}.
\end{theorem}

This theorem enables us to compute the $(K,\U(\I))$-local IS discrepancy function for each subsystem $A_i$. Although in the original definition we assume the IS discrepancy function is valid for any input signals $u_1,u_2 \in \U$, in practice $A_i$ can only take $A_j$'s outputs or states as inputs, which is bounded. Thus, \cite{HuangM:HSCC2014} can still use $(K,\U(\I))$-local IS discrepancy function computed by this approach. Furthermore, the $(K,\U(\I))$-local IS discrepancy function here can over-approximate the reachset of the systems in \eqref{eqn:system-input} with the input $u$ being chosen nondeterministically in some compact set.
 
\section{Experimental Evaluation}
\label{sec:experiments}
We have implemented the verification algorithm of Figure~\ref{fig:algo1} and the $\findldf$ subroutine both with and without coordinate transformation in Matlab. The implementation and the examples are available from~\cite{implementation}. For simulation we use Matlab's built-in ODE solver. 
The Jacobian matrix, an upper bound of the Lipschitz constant are given as inputs. In addition, the function to do the term-wise maximization of the error matrix is also given as inputs (see Section \ref{sec:analysis}).
We use the absolute error for ODE solver as the error bounds for simulation. 
%
The results presented here are based on experiments performed on an Intel Xeon V2 desktop computer.

\subsection{Comparison with other tools}
We compare the performance of our algorithm with two other tools, namely, Flow*~\cite{chen2013flow} and HyCreate~\cite{hycreate}, for safety verification problem of nonlinear dynamical systems. We use seven benchmarks which are shown in Table~\ref{tab:benchmarks} with time bound $T=10s$. 
Flow* uses Taylor models for approximating reachtubes from a set of initial states.
Currently, it returns ``Safe'' or ``Unknown'', but not ``Unsafe''. 
HyCreate uses the face-lifting approach of~\cite{dang1998reachability} and provides a intuitive interface for creating models. 

Vanderpol, CoupledVanderpol, JetEngine, and Brusselator are commonly used, low-dimensional, nonlinear benchmarks. Sinusoidal tracking \cite{sharma2008design} is a 6 dimensional nonlinear designed as a frequency estimator.
The  Lorenz Attractor (row 7) is a well known chaotic dynamical system. Robot arm is a 4 dimensional nonlinear system described in~\cite{angeli2000characterization}.
The Helicopter is a high dimension linear model of a helicopter system from~\cite{frehse2011spaceex}. 

We have implemented verification algorithmwith and without coordinate transformation.
Columns (\#SimO) and (LDFO(s)) show the number of simulations and running time of our algorithm (Figure \ref{alg:ComputeLDF}) without coordinate transformation. In comparison, Columns (\#Sim) and (LDF) are the results with coordinate transformation. Coordinate transformation provides tighter bounds, so the number of simulations and running time decrease under the same environment (i.e. same initial sets and unsafe sets). In row 10 and 11, we increase the time bound of the fixed-wing model to $T=50$ and $T=100$ respectively and the results show that the algorithm scales reasonably for longer time horizons. 
Flow* and HyCreate generate a single over-approximation of the reachtube from the initial set independent of the safety property. While our algorithm will refine the initial sets when the reachtube intersects with the unsafe set. In all of these benchmarks, we make the unsafe set close to the reachtubes, to make the models safe yet it needs a lot of refinements to arrive at that conclusion.
Overall, the proposed approach with coordinate transformation outperformed others in terms of the running time,  especially in high dimensional benchmarks. The ``N/A" in the table means the algorithm timed out at 30 minutes.
Of course, our implementation requires the users to give the symbolic expression of the Jacobian matrix and term-wise maximization functions, while Flow* and HyCreate just needs the differential equations. 
Moreover, our implementation currently handles only nonlinear dynamical systems, and both Flow* and HyCreate can handle hybrid systems.
\begin{sidewaystable}[htbp]
\small
  \centering
  \caption{Safety verification for benchmark examples. dim: dimension of the model; $\delta$: diameter of the initial set; $\mathbb{U}$: unsafe set; \#Sim: number of simulations with coordinate transformation; LDF: running time of our implementation (with coordinate transformation) in seconds; \#SimO: number of simulations using algorithm in Figure \ref{alg:ComputeLDF}; LDFO: running time of algorithm in Figure \ref{alg:ComputeLDF}(without coordinate transformation)in seconds.} 
    \begin{tabular}{c|c|c|c|c|c|c|c|c|c|c}
    \hline
    &example & dim & $\delta$ & $\mathbb{U}$ & \#Sim & LDF(s)  & \#SimO & LDFO(s) & flow*(s) & HyCreate(s) \\
    \hline
    1&Vanderpol & 2     & 0.5   & x$>$2.0 & 9     & 0.378 & 61& 2.01& 11.2  & 2.776 \\
    2&Brusselator & 2     & 0.5  & x$>$1.3 & 21    & 1.01  & 85& 2.79&11.8  & 1.84 \\
    3&Jet Engine & 2        & 0.4   & x$>$2.0 & 5    & 0.353 & 61 & 1.97& 8.74  & 5.54 \\
    4&Robot arm & 4     & 0.5   & x$>$2.5 & 81    & 4.66 &1159&47.9& 169   & $>$300 \\
    5&CoupledVanderpol & 4     & 0.5   & x$>$2.5 & 41    & 2.21  &1353&54.2& 93    & 49.8 \\
    6&Sinusoidal Tracking& 6 & 0.5 & x$>$10 & 185 & 13.2 &753& 97.0& 258&  $>$300\\
    7&Lorenz Attractor & 3 & 0.02& x$>$1e4 & 570 & 13.99 &3105&72.0& 53.4 & N/A \\
\hline
\hline

    8&Fixed-wing UAV (T=10)   & 7 & 3  & x$>$ 39 & 321 & 20.8 &N/A&N/A& N/A & N/A \\
    9&Helicopter & 28 & 0.02 & x$>$4 & 585 & 67.7 &N/A&N/A& N/A & N/A\\
    \hline
    \hline
    10&Fixed-wing UAV (T=50)   & 7 & 3  & x$>$ 39 & 321 & 99.8 &N/A&N/A& N/A & N/A \\
    11&Fixed-wing UAV (T=100)   & 7 & 3  & x$>$ 39 & 321 & 196&N/A&N/A & N/A & N/A \\
    \hline
    \end{tabular}%
  \label{tab:benchmarks}%
\end{sidewaystable}%

\subsection{Properties of LDF}
We explore the behavior of the algorithm with respect to changes in the relative positions of the initial set and the unsafe set. 
We use the nonlinear model of the Robot arm system. 
We fix the point $[1.5,1.5,0,0]$ as the center of the initial set and $T=10$ seconds as the time bound,
and vary the diameter of the initial set ($\delta$) and the unsafe set ($\mathbb{U}: \theta > c$), where $\theta$ is the angle of the arm. 
The number of simulations used by the algorithm with coordinate transformation (\#Sim),
the diameter of the reach tube at the final time $T$  (dia), and 
the total running time (RT) are shown in Table~\ref{tab:property}. 

From the first 5 rows in the Table, we see the expected behavior that for a fixed unsafe set, the diameter of the Reachtube decreases with decreasing $\delta$.
This corresponds to the property that the discrepancy function $\beta(x,x',t)$ goes to $0$ as the initial points $x \rightarrow x'$, and therefore the error in the reachability computation decreases monotonically with the  
diameter of the initial set. Rows $4$ and $6$-$9$ show that if we fix the size of the initial set, then as the unsafe set comes closer to the actual reachtube, the number of simulations increases and therefore the running time increases until the system becomes unsafe. 
As more refinements are made by the algorithm, the accuracy (measured by the diameter of the reachtube) improves.
Similar trend is seen in rows $10$-$12$, the algorithm will need more refinements to find a counter example that shows unsafe behavior, if the unsafe set is close to the boundary of the reachtube.
\begin{table}[htbp]
\small
  \centering
  \caption{Safety verification for a robot arm with different initial states and unsafe sets. safety: safety result returned by verification algorithm;}
    \begin{tabular}{c|c|c|c|c|c|c}
    \toprule
     & $\delta$ & $\mathbb{U}$ & saftey & \#Sim & dia & RT(s) \\
    \hline
     1 &  0.6   & $\theta >$3     & safe  & 17    & 5.6e-3 & 0.948 \\
     2 &  0.4   & $\theta >$3     & safe  & 9     & 3.6e-3 & 0.598 \\
     3 &  0.3   & $\theta >$3     & safe  & 9     & 2.6e-3 & 0.610 \\
     4 &  0.2   & $\theta >$3     & safe  & 5     & 1.8e-3 & 0.444 \\
     5 &  0.1   & $\theta >$3     & safe  & 1     & 1.5e-3 & 0.271 \\
     6 &  0.2   & $\theta >$2.5   & safe  & 9     & 1.7e-3 & 0.609 \\
     7 &  0.2   & $\theta >$2.18  & safe  & 17    & 1.4e-3 & 0.933 \\
     8 &  0.2   & $\theta >$2.17  & safe  & 29    & 1.0e-3 & 1.429 \\
     9 &  0.2   & $\theta >$2.15  & safe  & 161   & 9.2e-4 & 6.705 \\
     10 & 0.2   & $\theta >$2.14  & unsafe & 45    & N/A   & 1.997 \\
     11 & 0.2   & $\theta >$2.13  & unsafe & 35    & N/A   & 1.625 \\
     12 & 0.2   & $\theta >$2.1   & unsafe & 1     & N/A   & 0.267 \\
    \bottomrule
    \end{tabular}%
  \label{tab:property}%
\end{table}%

Next, we explore the behavior of the algorithm (with coordinate transformation) with large initial sets. We use the $7$ dimensional model of a fixed-wing UAV. 
The initial sets are defined as balls with different radii around a center point $[30,980,0,125,0,0,30.4]$ and $\delta$ in the first column is the diameter of the initial sets. The unsafe set is defined as $H>c$, where $H$ is the thrust of UAV.
The time horizon is fixed at $T =10$ seconds. As shown in Table~\ref{tab:initial-size}, our algorithm can handle large initial set and high dimension systems. Although it may need many simulations (24001 covers), the algorithm terminates in 30 mins.
All the results of this table are safe.
\begin{table}[htbp]
\small
  \centering
  \caption{Safety verification for a fixed-wing UAV with large initial sets.}
    \begin{tabular}{c|c|c|c|c}
    \toprule
    &$\delta$ & $\mathbb{U}$ & \#Sim & RT(s) \\
    \hline
    1&50    & $H>400$  & 24001 & 1518 \\
    2&46    & $H>400$  & 6465  & 415 \\
    3&40    & $H>400$  & 257   & 16.33 \\
    4&36    & $H>400$  & 129   & 8.27 \\
    5&20    & $H>400$  & 1     & 0.237 \\
    \bottomrule
    \end{tabular}%
  \label{tab:initial-size}%
\end{table}%
\section{Related Work}\label{sec:related}
\label{sec:related}

Simulation based verification has been studied in several papers recently \cite{donze2007systematic,althoff2008reachability,dang2008sensitive,julius2009trajectory}. In \cite{donze2007systematic} the authors introduce a general simulation based method for proving safety of arbitrary continuous systems. The novelty of their approach consist in the use of sensitivity analysis, where the {\em sensitivity matrix \/} with respect to initial state $x_0$ at time $t$ is defined as $s_{x_0} \triangleq \frac{\partial \xi(x_0,t)}{\partial x_0}$.  It is shown that $\dot s_{x_0}(t) = J_f(x_0,t)s_{x_0}(t)$ and $s_{x_0}(t)$ can be solved by efficient solvers. Then $\|s_{x_0}(t)\|\delta$ is used to bound the distance $\|\xi(x,t)-\xi(x_0,t)\|$ for $x \in B_{\delta}(x_0)$ at time $t$. It is shown that this upperbound holds for linear time varying systems. For general nonlinear systems, $\|s_{x_0}(t)\|\delta$ has a quadratic error term with respect to $\delta$ that requires further analysis.  Thus, this technique is sound for linear system but does not provide any formal guarantees for nonlinear systems (\cite{donze2007systematic}, page 13). In \cite{dang2008sensitive} this technique is extended to nonlinear systems subject to disturbances as inputs and uncertainty in the initial conditions to obtained an approximation that ignores the higher order terms. In contrast, in Section \ref{sec:DF} and Section \ref{sec:ISDF} we have provided a strict over-approximation of Lipschitz continuous systems with respect to uncertainty in the initial conditions and uncertainty in the input signals. In \cite{julius2009trajectory}, the authors provide several approaches to capture the upperbound of the distance between two trajectories for linear systems and some polynomial systems.

In \cite{han2013towards} the authors present a convenient implementation of sensitivity analysis in the Simulink software. Again, the trajectory sensitivity matrix can only be used as a linear approximation for a {\em perturbed trajectory\/} , instead of over-approximation of the reachset. In \cite{althoff2008reachability} the authors provide a different approach by linearizing the nonlinear system locally, and bounding the linearization error by Lagrange remainders. The original definition of discrepancy function can be seen as a generalization of the incremental stability \cite{angeli2002lyapunov}. The incremental Lyapunov  function  can be used as discrepancy function when a system is incrementally stable. An incremental Lyapunov function-based approach is used in \cite{girard2010approximately}. Here the authors go much further and construct a finite symbolic model that is approximately bisimilar to the original switched system. Our approach bypasses the incremental stability requirement by focusing on bounded time analysis.

Contraction  in \cite{slotine98} is defined as the region in which the eigenvalues of the symmetric part of the Jacobian is uniformly negative. The authors use ``virtual displacement" to get the result, while we get the upperbound of the eigenvalues of the symmetric part of the Jacobian directly from the generalized mean value function. Contraction metrics introduced in \cite{slotine98} is also used in \cite{DMVemsoft2013} to perform sound and relative complete analysis of nonlinear systems.




\section{Conclusions and Future Work}
In this paper, we present an algorithm $\findldf$ to compute local discrepancy functions, which is an upperbound of the distance between trajectories starting from an initial set. The algorithm computes the rate of trajectory convergence or divergence for small time intervals and gives the rate as coefficients of a continuous piece-wise exponential function. The local discrepancy we compute satisfies the definition of discrepancy function, so the verification algorithm using $\findldf$ as a subroutine is sound and relatively complete. We also provide a coordinate transformation method to improve the estimation of rates. Furthermore, we extend the algorithm to compute input-to-state discrepancy functions. 

In the future, we plan on using more rigorous ODE solvers like \cite{capd} and embedding the algorithm in verification tools like C2E2 \cite{DMVemsoft2013} for safety verification of hybrid systems.


\bibliographystyle{abbrv}
\bibliography{citation,sayan1}

\begin{thebibliography}{10}

\bibitem{althoff2008reachability}
M.~Althoff, O.~Stursberg, and M.~Buss.
\newblock Reachability analysis of nonlinear systems with uncertain parameters
  using conservative linearization.
\newblock In {\em CDC 2008. 47th IEEE Conference on}, pages 4042--4048. IEEE,
  2008.

\bibitem{angeli2002lyapunov}
D.~Angeli.
\newblock A lyapunov approach to incremental stability properties.
\newblock {\em IEEE Transactions on Automatic Control}, 47(3):410--421, 2002.

\bibitem{angeli2000characterization}
D.~Angeli, E.~D. Sontag, and Y.~Wang.
\newblock A characterization of integral input-to-state stability.
\newblock {\em Automatic Control, IEEE Transactions on}, 45(6):1082--1097,
  2000.

\bibitem{annpureddy2011s}
Y.~Annpureddy, C.~Liu, G.~Fainekos, and S.~Sankaranarayanan.
\newblock {\em S-taliro: A tool for temporal logic falsification for hybrid
  systems}.
\newblock Springer, 2011.

\bibitem{capd}
CAPD.
\newblock Computer assisted proofs in dynamics.
\newblock url{http://www.capd.ii.uj.edu.pl/}, 2002.

\bibitem{chen2013flow}
X.~Chen, E.~{\'A}brah{\'a}m, and S.~Sankaranarayanan.
\newblock Flow*: An analyzer for non-linear hybrid systems.
\newblock In {\em CAV}, pages 258--263. Springer, 2013.

\bibitem{dang2008sensitive}
T.~Dang, A.~Donz{\'e}, O.~Maler, and N.~Shalev.
\newblock Sensitive state-space exploration.
\newblock In {\em CDC 2008. 47th IEEE Conference on}, pages 4049--4054. IEEE,
  2008.

\bibitem{dang1998reachability}
T.~Dang and O.~Maler.
\newblock Reachability analysis via face lifting.
\newblock In {\em HSCC}, pages 96--109. Springer, 1998.

\bibitem{donze2007systematic}
A.~Donz{\'e} and O.~Maler.
\newblock Systematic simulation using sensitivity analysis.
\newblock In {\em HSCC}, pages 174--189. Springer, 2007.

\bibitem{DMVemsoft2013}
P.~S. Duggirala, S.~Mitra, and M.~Viswanathan.
\newblock Verification of annotated models from executions.
\newblock In {\em Proceedings of the Eleventh ACM International Conference on
  Embedded Software}, page~26. IEEE Press, 2013.

\bibitem{DuggiralaWMVM14}
P.~S. Duggirala, L.~Wang, S.~Mitra, M.~Viswanathan, and C.~Mu{\~n}oz.
\newblock Temporal precedence checking for switched models and its application
  to a parallel landing protocol.
\newblock In {\em FM 2014}, pages 215--229. Springer, 2014.

\bibitem{implementation}
C.~Fan and S.~Mitra.
\newblock Bounded verification with on-the-fly discrepancy computation (full
  version).
\newblock available at
  \url{http://web.engr.illinois.edu/~cfan10/research.html}.

\bibitem{frehse2011spaceex}
G.~Frehse, C.~Le~Guernic, A.~Donz{\'e}, S.~Cotton, R.~Ray, O.~Lebeltel,
  R.~Ripado, A.~Girard, T.~Dang, and O.~Maler.
\newblock Spaceex: Scalable verification of hybrid systems.
\newblock In {\em CAV}, pages 379--395. Springer, 2011.

\bibitem{girard2010approximately}
A.~Girard, G.~Pola, and P.~Tabuada.
\newblock Approximately bisimilar symbolic models for incrementally stable
  switched systems.
\newblock {\em Automatic Control, IEEE Transactions on}, 55(1):116--126, 2010.

\bibitem{han2013towards}
Z.~Han and P.~J. Mosterman.
\newblock Towards sensitivity analysis of hybrid systems using simulink.
\newblock In {\em HSCC}, pages 95--100. ACM, 2013.

\bibitem{HuangFMMK14}
Z.~Huang, C.~Fan, A.~Mereacre, S.~Mitra, and M.~Z. Kwiatkowska.
\newblock Invariant verification of nonlinear hybrid automata networks of
  cardiac cells.
\newblock In {\em CAV 2014.}, pages 373--390. Springer, 2014.

\bibitem{HuangM:HSCC2014}
Z.~Huang and S.~Mitra.
\newblock Proofs from simulations and modular annotations.
\newblock In {\em In 17th International Conference on Hybrid Systems:
  Computation and Control}, Berlin, Germany. ACM press.

\bibitem{julius2009trajectory}
A.~A. Julius and G.~J. Pappas.
\newblock Trajectory based verification using local finite-time invariance.
\newblock In {\em HSCC}, pages 223--236. Springer, 2009.

\bibitem{slotine98}
W.~Lohmiller and J.-J.~E. Slotine.
\newblock On contraction analysis for non-linear systems.
\newblock {\em Automatica}, 34(6):683--696, 1998.

\bibitem{vnode2006}
N.~Nedialkov.
\newblock {VNODE-LP}: Validated solutions for initial value problem for {ODE}s.
\newblock Technical report, McMaster University, 2006.

\bibitem{sharma2008design}
B.~B. Sharma and I.~N. Kar.
\newblock Design of asymptotically convergent frequency estimator using
  contraction theory.
\newblock {\em Automatic Control, IEEE Transactions on}, 53(8):1932--1937,
  2008.

\bibitem{hycreate}
B.~Stanley and C.~Marco.
\newblock Computing reachability for nonlinear systems with hycreate.
\newblock In {\em Demo and Poster Session, HSCC}.

\end{thebibliography}

\newpage
\appendix
\section{Appendix: Proofs of Lemmas}
Proof of Lemma \ref{lemma:highdimensionmean}:

In this proof, the $i$'s in subscript correspond the the $i^{th}$ components of the vector functions.
For any $t \in [0,1], i \in \{1,\dots,n\}$, we define  $g_i(t):=f_i(x+tr)$.
Then we have 
\begin{eqnarray}\label{eqn:int}
f_i(x+r)  - f_i(x) &=& g_i(1)-g_i(0) = \int_{0}^{1}{\frac{dg_i(t)}{dt}dt}. 
\end{eqnarray}
Using the chain rule of gradient, we have
\begin{eqnarray} \label{eqn:derivative}
\frac{dg_i(t)}{dt}&=&\left. \nabla f_i(u)\right| _{u=x+tr} \cdot \frac{d (x+tr)}{d t} \nonumber \\
&=&\left. \nabla f_i(u)\right| _{u=x+tr} \cdot r= \sum_{j=1}^{n}{\left.\frac{\partial f_i(u)}{\partial u_j}\right|_{u=x+tr}r_j},
\end{eqnarray}
where $\nabla f_i (u) = [\frac{\partial f_i(u)}{\partial u_1},\frac{\partial f_i(u)}{\partial u_2},\dots,\frac{\partial f_i(u)}{\partial u_n}]$ is the gradient of function $f_i$. Substituting \eqref{eqn:derivative} in \eqref{eqn:int}, we have:
\begin{eqnarray*}
f_i(x+r)  - f_i(x) &=&  \int_{0}^{1}{\left(\sum_{j=1}^{n}{\left.\frac{\partial f_i(u)}{\partial u_j}\right|_{u=x+sr}r_j}\right)} ds \nonumber \\
&=& \sum_{j=1}^{n}{\left( \int_{0}^{1}{\left.\frac{\partial f_i(u)}{\partial u_j}\right|_{u=x+sr}ds} \right) r_j}. \nonumber \\
\end{eqnarray*}
Since $J_f(x+sr)$ is the matrix consisting of the components of $\left.\frac{\partial f_i(u)}{\partial u_j}\right|_{u=x+sr}$, the lemma holds.

Proof of Theorem \ref{thm:lamd_ordering}.

This theorem is established by the minimax characterization of the eigenvalues. Let $\tilde A = A+E$, and let $\lambda_i(A), \lambda_i(E),\lambda_i(\tilde A)$ denote the eigenvalues of $A,E$ and $\tilde A$ respectively , where all three sets are arranged in non-increasing order. By the maxmin therorem we have 
$$
\lambda _k(\tilde A) = \min_{\dim \mathcal{V}=n-k+1} \left( \max_{0 \neq v \in \mathcal {V}} \rho_{\tilde A} (v)\right)
$$
Which can also be written as

\begin{eqnarray*}
\lambda_k(\tilde A)=\min \max (x^T \tilde A x) \\
x^T x =1, p_i^T x =0 (i =1,2,3,\dots,k-1)
\end{eqnarray*}
Hence, if we take any particular set of $p_i$, we have for all corresponding $x$,
\begin{equation}\label{inequation for sum of two matrix 1}
\lambda_k (\tilde A) \leq \max(x^T \tilde A x) = \max (x^T A x + x^T E x).
\end{equation}

If $U^T A U = \Lambda = \mbox{diag}(\lambda_i(A))$ and $U$ is the orthogonal matrix, then if we take $p_i = Ue_i$ the relations to be satisfied are
$$
0 = p_i^T x = e_i^Ty (i = 1,2,\dots, k-1)
$$

With this choice of the $p_i$ then the first $k-1$ components of $y$ are zero, and from equation \eqref{inequation for sum of two matrix 1} we have
\begin{equation}\label{inequation for sum of two matrix 2}
\lambda_k(\tilde A) \leq \max (x^T A x + x^T E x) \leq \max (\sum_{i=k}^{n}{\lambda_i(A)y_i^2} + x^T E x)
\end{equation}
However,
\begin{equation}
\sum_{i=k}^{n}{\lambda_i(A) y_i^2} \leq \lambda_k(A)
\end{equation}
while
\begin{equation}
x^T E x \leq \lambda_1(E)
\end{equation}
for any $x$.
Hence the expression in brackets of equation \eqref{inequation for sum of two matrix 2} is not greater than $\lambda_k(A)+\lambda_1(E)$ for any $x$ corresponding to this choice of the $p_i$. Therefore its maximum is not greater than $\lambda_k(A)+\lambda_1(E)$ and we have
\begin{equation}\label{inequation for sum of two matrix 3}
\lambda_k(\tilde A) \leq \lambda_k(A)+\lambda_1(E)
\end{equation}
Since $A = \tilde A + (-E)$ and the eigenvalues of $-E$ in non-increasing order are $-\lambda_n(E),-\lambda_{n-1}(E),\dots,-\lambda_1(E)$, and application of the result we have just proved gives
\begin{equation}\label{inequation for sum of two matrix 4}
\lambda_k(A) \leq \lambda_k(\tilde A) + (-\lambda_n (E))~~~
\mbox{or}~~~
\lambda_k(\tilde A) \geq \lambda_k(A) + \lambda_n(E)
\end{equation}

Thus we have
$$
\lambda_k(A)+\lambda_n(E) \leq \lambda_k(A+E) \leq \lambda_k(A)+\lambda_1(E)
$$
The relations \eqref{inequation for sum of two matrix 3} and \eqref{inequation for sum of two matrix 4} imply that when $E$ is added to $A$ all of its eigenvalues are changed by an amount which lies between the smallest and greatest  of the eigenvalues of $E$. Note that we are not concerned here specifically with small perturbations and the results are not affected by multiplicities in the eigenvalues of $A,E$ and $A+E$.


\end{document}